\documentclass[runningheads]{llncs}
\usepackage[T1]{fontenc}
\usepackage{graphicx}
\usepackage{natbib}
\usepackage{algorithm}
\usepackage{algpseudocode}
\usepackage{amssymb}
\usepackage{amsmath}
\usepackage{mathabx}
\usepackage{mathtools}
\usepackage{listings}
\usepackage{caption}
\usepackage{float}
\usepackage{tcolorbox}
\usepackage{pgfplots}

\newcommand{\Ac}{\mathcal{A}}
\newcommand{\RR}{\mathcal{R}}
\newcommand{\Om}{\Omega}
\newcommand{\Gval}{V}

\newcommand{\SOL}{\mathcal{S}}

\newcommand{\NE}{\mathrm{NE}}

\newcommand{\sDSE}{\mathrm{sDSE}}
\newcommand{\DSE}{\mathrm{DSE}}
\newcommand{\wDSE}{\mathrm{wDSE}}
\newcommand{\sNE}{\mathrm{sNE}}
\newcommand{\pNE}{\mathrm{pNE}}
\newcommand{\CE}{\mathrm{CE}}
\newcommand{\CCE}{\mathrm{CCE}}
\newcommand{\RRo}{\mathcal{R}^{o}}
\newcommand{\QQ}{\mathcal{Q}}
\newcommand{\Ad}{\mathbf{a}^{\dagger}}
\newcommand{\ad}{a^{\dagger}}
\newcommand{\ac}{\mathbf{a}}

\newcommand{\Az}{\textsc{A-Zero}}
\newcommand{\Ao}{\textsc{A-General}}
\newcommand{\Azip}{\textsc{A-Zero}}
\newcommand{\Aoip}{\textsc{A-General}}
\newcommand{\Azdp}{\textsc{A-Zero}}
\newcommand{\Aodp}{\textsc{A-General}}

\newcommand{\BigOh}[1]{\mathcal{O}\left({#1}\right)}
\newcommand{\ints}[1]{\left[{#1}\right]}
\newcommand{\card}[1]{\left|{#1}\right|}

\spnewtheorem{thm}{Theorem}{\bfseries}{\itshape}
\spnewtheorem{cor}[thm]{Corollary}{\bfseries}{\itshape}
\spnewtheorem{lem}[thm]{Lemma}{\bfseries}{\itshape}
\spnewtheorem{prop}[thm]{Proposition}{\bfseries}{\itshape}
\spnewtheorem{conj}[thm]{Conjecture}{\bfseries}{\itshape}
\spnewtheorem{obs}[thm]{Observation}{\bfseries}{\itshape}
\spnewtheorem{clm}[thm]{Claim}{\bfseries}{\itshape}

\spnewtheorem{algo}{Algorithm}{\bfseries}{\itshape}

\spnewtheorem{df}[thm]{Definition}{\bfseries}{\rmfamily}
\spnewtheorem{eg}[thm]{Example}{\bfseries}{\rmfamily}
\spnewtheorem{asm}[thm]{Assumption}{\bfseries}{\rmfamily}
\spnewtheorem{cond}[thm]{Condition}{\bfseries}{\rmfamily}

\spnewtheorem{rmk}[thm]{Remark}{\itshape}{\rmfamily}

\allowdisplaybreaks

\begin{document}
\title{The Game Changer Problem: Controlling Equilibria with Discrete Rewards}
\titlerunning{The Game Changer Problem}
%
\author{Brandon Han \and Young Wu \and Shiyun Cheng \and Xiaojin Zhu}
\authorrunning{Brandon Han et al.}
%
\institute{University of Wisconsin - Madison}
%
\maketitle              
\begin{abstract}
We introduce the game changer problem, where an external designer modifies a game’s reward matrix to make a target pure action profile the unique equilibrium, subject to the constraint that all entries of the reward matrix come from a finite set. We give simple feasibility characterizations for two-player zero-sum games and general-sum games, and the discrete reward structure yields exact optimality and enables efficient dynamic programming algorithms, providing a sharper alternative to prior continuous reward redesign formulations based on linear programming.

\end{abstract}

\section{Introduction} 
We consider the task of controlling joint player behaviors by changing payoffs. In many real multi-agent systems, payoffs cannot be arbitrary real numbers; they must come from a small, meaningful set. This discreteness creates a new algorithmic problem and, in fact, enables exact optimal redesign via efficient dynamic programming rather than generic linear-programming formulations. In our setup, there is a new system administrator (the game changer), and $n$ strategic agents (players), for example, a company with a new CEO and $n$ employees, a sports competition with a new tournament organizer and $n$ competitors, or even a kingdom with a new king and $n$ subjects. The existing system is operated as an $n$-player general-sum game with reward function $\RRo$, set by the previous administrator and played by the agents. $\RRo_{i}\left(\ac_i, \ac_{-i}\right)$ is the payoff to agent $i \in \left[n\right]$ with joint action $a = \left(\ac_i, \ac_{-i}\right)$. Suppose it is known that the players are rational and seek a Nash Equilibrium $\left(\NE\right)$. While a possibly mixed $\NE$ is guaranteed to exist, there could be many different NEs, so it is still hard for the new administrator to predict how the game will actually play out. In addition, the new administrator may have a different value $V\left(\ac\right)$ for a joint action $\ac$ than the "social welfare" $\displaystyle\sum_{i=1}^{n} \RR_i\left(\ac\right)$. For example, the administrator may prefer the joint action in which all agents work together on a larger project rather than the current equilibrium, in which the agents work on smaller projects on their own. The NEs under the old game $\RRo$ may have low values. Thus, the new administrator wants to reform the system (or change the game, hence the name game changer). More importantly, our paper introduces a novel consideration when changing the game, which is that the new game rewards only contain entries from a finite set. In our system administrator story, entries could be restricted to $5$ valid values, say, $A, B, C, D, F$ for an agent's evaluation score. In the other examples, the discrete-valued entries represent integer hourly wages in the company environment, a finite number of ratings in sports competitions, or the number of gold coins in the kingdom example. The discrete reward constraint is not merely a modeling restriction: because finite sets have a minimum positive gap between distinct values, strict inequalities can be handled exactly by adjacent reward levels, and this avoids the relaxations in prior continuous formulations and allows us to compute the exact optimal redesigned games.

Concretely, the $n$ players and their action spaces are kept the same, but the game changer changes the payoff function $\RRo$ to $\RR$ with the following considerations,
\begin{enumerate}
\item What actions the players will choose under $\RR$ should be highly predictable. For example, it would be nice if $\RR$ has a unique pure $\NE$ $\Ad$.
\item What action profile the players use is valuable to the game changer; for example, $V\left(\Ad\right)$ is high in the unique pure $\NE$.
\item The game changer may have preferences on $\RR$ itself regardless of the players' behaviors. For example, the reform is palatable if $\left\|\RR - \RRo\right\|_{1}$ is small, $\RR$ should contain only entries from a discrete set, and in fact, one focus of the present paper is an extreme form of the latter: $\RR$ must take value in a given discrete (and potentially finite) set $\Omega$.
\end{enumerate}

Our problem is adjacent to mechanism design and implementation theory. Our problem does not involve players with private types, and rather than designing message spaces, allocation rules, or transfers, we keep the players and action spaces fixed and directly redesign the reward matrices. We also do not consider the game changer as one of the players due to its asymmetric role in the problem,
\begin{itemize}
\item The game changer's action set is qualitatively different and includes all game reward functions $\RR$ with some fixed action space $\Ac$, which is not an ordinary move like a finite set $\Ac_i$ for player $i$,
\item The game changer has full commitment power, and there is no strategic uncertainty from the players' perspective. In particular, even in the Markov game setting with sequential interactions, the game changer cannot adaptively change the rules of the game.
\end{itemize}

The main contributions of this paper are,
\begin{itemize}
\item We formulate the discrete game-changer problem, where an external game changer redesigns a reward matrix using only values from a prescribed finite set $\Omega$, with the goal of inducing a unique target equilibrium while minimizing the deviation from the original game.
\item We provide necessary and sufficient feasibility conditions: $|\Omega| \geq 3$ for zero-sum games, for installing a unique strict pure Nash equilibrium, which also yields a unique Nash equilibrium, correlated equilibrium, or coarse correlated equilibrium; and $|\Omega| \geq 2$ for general-sum games, for installing a strict dominant strategy equilibrium.
\item We derive an exact dynamic-programming algorithm for the discrete game-changer problem, avoiding generic integer/linear programming formulations, thereby efficiently computing the optimal redesigned reward matrix. Unlike prior continuous reward-redesign formulations, the discrete setting gives exact optimality rather than approximate optimality through strict-inequality relaxations.
\end{itemize}

It is also easy to see the security relevance of this task: an adversarial attacker may play the role of the administrator (game changer) and hack the payoffs to control agent (player) behavior, and our work exposes the quantitative game-theoretic threat. This will enable defenders to target specifically these types of attacks. For instance, by diligently checking and catching small changes to $\RRo$, the defender effectively imposes a high cost on $\left\|\RRo - \RR\right\|$ and may force the attacker to give up because $\RR = \RRo$ will be the optimal solution to the reform. But a benevolent policymaker can also benefit from the present paper by computing the optimal new policy $\RR$ to encourage beneficial behaviors (for example, all players volunteer in the volunteer game).

\section{Related Work} 

Our work is the closest to the mechanism design literature, including \cite{maskin2002implementation, moore1990nash}. Classical implementation theory papers focus on implementing a desired social choice rule under private information. Our problem is much narrower but more concrete. The design variable is the new reward matrix, not a general mechanism with message spaces, and the problem does not involve transfers, as in the k-implementation problem in \cite{monderer2003k}. Moreover, we require uniqueness of the equilibrium, unlike many mechanism design problems that rely on indifference.

This paper is a continuation of the line of work focused on game redesign and reward engineering, including \cite{eshoa2024precision, ma2021game, wu2023minimally, mcmahan2025optimally}. The novel modeling choice is a finite or discrete set of rewards, $\Omega$. Given this realistic assumption, we provide a simple feasibility condition and an efficient dynamic programming algorithm that computes the exact optimal reward matrix, whereas in previous papers the optimization is deferred to a linear program, yielding only approximately optimal solutions due to the relaxation of strict inequality constraints.

The algorithms for game redesign have also been applied to adversarial reinforcement learning  \cite{rakhsha2021policy, zhang2020adaptive, rakhsha2021reward, xu2024reward}, and adversarial multi-agent reinforcement learning \cite{wu2023reward, wu2024data, mcmahan2024optimal, liu2023efficient}. In this literature, reward poisoning manipulates training data or reward feedback to induce policies, whereas our work focuses on directly modifying the payoff matrices to install a target equilibrium behavior under the discrete-set entry constraints.

\section{Problem Setup} 
We call the administrator's reform "the game changer problem" because it changes the old game $\RRo$ to a new game $\RR$. We define the general form of this problem below.

First we define the players by the tuple $\left(n, \left\{\Ac_i\right\}_{i=1}^{n}, \SOL\right)$ where $\SOL$ is the solution concept they seek, for example, Nash Equilibrium $\left(\NE\right)$, Dominant Strategy Equilibrium $\left(\DSE\right)$, Correlated Equilibrium $\left(\CE\right)$, and so on. Given a game $\RR$, we also overload notation and let $\SOL\left(\RR\right)$ denote the set of equilibria under that solution concept. For example, $\SOL\left(\RR\right)$ can be the set of NEs (pure and mixed) of $\RR$.

We characterize the game changer's three considerations with three functions, which are inputs to our problem,
\begin{enumerate}
\item $P\left(\SOL\left(\RR\right)\right)$ is the predictability of player behaviors. A simple example is $P\left(\SOL\left(\RR\right)\right) = -\left| \SOL\left(\RR\right) \right|$, the negative cardinality, for instance, the number of distinct NEs. The negation is to make predictability high when the solution is simple;
\item $V\left(\SOL\left(\RR\right)\right)$ is the extrinsic value to the game changer on how the game will be played out. If $\SOL\left(\RR\right)$ is a unique pure $\NE$ with joint action $\ac, V\left(\SOL\left(\RR\right)\right) = V\left(\ac\right)$ specifies the value to the game changer when the players are essentially forced to play the only rational option $\ac$. For example, $\ac_1 = \ac_2 = ... = \ac_{n} = \text{\;volunteer\;}$ may have higher value to the game changer. Note $V$ is an input to our problem, and is in general distinct from the sum of payoffs to the players $\displaystyle\sum_{i=1}^{n} \RR_{i}\left(\ac\right)$. This flexibility sets us apart from the standard notion of social welfare maximization. When $\SOL\left(\RR\right)$ has multiple equilibria, $V$ can be the average or worst-case value of these members.
\item $C\left(\RR, \RRo\right)$ is the cost of installing the new game $\RR$, potentially with reference to the old game $\RRo$. For example, $C\left(\RR, \RRo\right) = \left\|\RR - \RRo\right\|_{1.}$
\end{enumerate}

With these definitions, the game-changer problem is the following optimization problem,
\begin{align}
&\RR^\star \in \mathop{\mathrm{argmax}}_{\RR \in \mathbb{R}^{\times_{i=1}^{n} \left| \Ac_i \right|}} P\left(\SOL\left(\RR\right)\right) + V\left(\SOL\left(\RR\right)\right) - C\left(\RR, \RRo\right) \label{eq:ggc}
\end{align}
The general problem ~\eqref{eq:ggc} is powerful but difficult to solve for arbitrary $P, V, C$ functions. In this paper, we focus on a special yet meaningful case often encountered in applications.

We impose two restrictions,
\begin{enumerate}
\item The new game $\RR$ must induce a unique pure strategy solution: $\SOL\left(\RR\right) = \left\{\Ad\right\}$. For example, the new game must have a unique pure $\NE$. It does not matter how many solutions the old game has. This restriction is highly desirable to the game changer: the game changer will know precisely how the players will behave under $\RR$. We may stylize the predictability function as an infinite indicator such that $P\left(\SOL\left(\RR\right)\right) = 0$ if $\SOL\left(\RR\right)$ is a unique pure joint action, and $-\infty$ otherwise.
\item The entries in $\RR$ must be in a given discrete set $\Omega$. For example, in win-lose-tie games such as rock-paper-scissors, it is natural to map each outcome to a different number so $\Omega = \left\{-1, 0, 1\right\}$. It is less natural (though not impossible) to allow entries of $\RR$ to take arbitrary real values. Another justification stems from the security angle, where an attack with "strange" payoff values after reward poisoning may be easily detected, but an attack that stays within $\Omega$ can be hard to detect. We may stylize the cost function as $C\left(\RR, \RRo\right) = \infty$ if some elements in $R \neq \Omega$, and $\left\|\RR - \RRo\right\|_{1}$ otherwise. We consider finite sets $\Omega$, but in the appendix, we explain how problems with $\Omega = \mathbb{Z}$ are equivalent to problems with finite $\Omega$.
\end{enumerate}

Under these restrictions, we arrive at the Fully Predictable Discrete Game Change Problem defined by the tuple $\left(\RRo, \Omega, n, \left\{\Ac_i\right\}_{i=1}^{n}, \SOL\right)$. Since the solution to $\RR$ must be a unique joint action $\Ad$, we can first enumerate (pure) joint actions $\Ad$, then, conditioned on it, find the optimal $\RR$ with discrete entries. Furthermore, we can separate the cost function $C\left(\RR, \RRo\right)$ into a hard discrete constraint and a soft distance-to $\RRo$ norm, that is,
\begin{align}
\displaystyle\max_{\Ad} \displaystyle\max_{\RR} & \, V\left(\Ad\right) - \left\|\RR - \RRo\right\|_1 \label{eq:fpdgc}
\\ s.t. & \,\SOL\left(\RR\right) = \left\{\Ad\right\}\nonumber
\\ & \,\RR\left(\ac\right) \in \Omega, \forall\; \ac \in \Ac\nonumber
\end{align}
The optimization ~\eqref{eq:fpdgc} for the Fully Predictable Discrete Game Changer Problem has a difficult hard constraint $\SOL\left(\RR\right) = \left\{\Ad\right\}$, namely the solution to $\RR$ must be the unique pure strategy $\Ad$. Our approach depends on whether $\RR$ should be a two-player zero-sum game or $n$-player general-sum game,
\begin{itemize}
\item If $\RR$ is a two-player zero-sum game, we aim to install a strict $\NE$ at $\Ad$.
\item If $\RR$ is an $n$-player general-sum game, we aim to install a strict $\DSE$ at $\Ad$.
\end{itemize}

In the following sections, we decompose \eqref{eq:fpdgc} into independent inner problems indexed by $\ad$, and focus on efficient algorithms to solve the inner problem exactly; the outer problem can be solved by enumeration over the (finite number of) pure action profiles.

\section{Game Changer on Two-Player Zero-sum Games}

Our first setting focuses on two-player zero-sum games. Here, we require both the given reward matrix $\RRo$ and the output matrix $\RR$ to be zero-sum. Our strategy is to iterate through every pure strategy $\Ad \in \Ac$ and to install each $\Ad$ as a strict Nash equilibrium. 

\begin{df} [Strict Nash Equilibrium] \label{eg:sne} 
    Let $\RR$ be the reward matrix of a zero-sum normal form game with action set $\Ac = \Ac_1 \times \Ac_2$. A pure joint strategy $\Ad \in \Ac$ is a strict Nash Equilibrium ($\sNE$) if for every $a_1 \in \Ac_1$ and $a_2 \in \Ac_2$, we have,
    \begin{align*}
        & a_1 \neq \ad_1 \implies \RR\left(a_1, \ad_2\right) < \RR\left(\ad_1, \ad_2\right) \\
        & a_2 \neq \ad_2 \implies \RR\left(\ad_1, a_2\right) > \RR\left(\ad_1, \ad_2\right).
    \end{align*}
\end{df}

Mathematically, we are solving the following minimization equation. 
\begin{align}
    \max_{\Ad} \max_{\RR} &\, \Gval\left(\Ad\right) - \left\|\RR - \RRo\right\|_1 \label{eq:framework-a0} \\
    \text{s.t.} 
    &\,\RR\left(a_1, \ad_2\right) < \RR\left(\ad_1, \ad_2\right) \quad \forall a_1 \neq \ad_1 \nonumber \\
    &\, \RR\left(\ad_1, a_2\right) > \RR\left(\ad_1, \ad_2\right) \quad \forall a_2 \neq \ad_2 \nonumber\\
    &\, \RR\left(\ac\right) \in \Om \quad \forall \ac \in \Ac \nonumber
\end{align}

This approach has several advantages. One such advantage is that the only additional constraints are linear inequalities, allowing us to directly apply integer programming. Moreover, every inequality involves the term $\RR\left(\Ad\right)$. This additional structure allows us to develop an efficient, closed-form algorithm for finding optimal solutions. We discuss this algorithm in Section~\ref{sec:alg-a0}. 

Another advantage is feasibility. We show a sufficient and necessary condition for the existence of a feasible solution to \eqref{eq:framework-a0}. We assume throughout this section and the next that each player has at least $2$ actions since degenerate one-action cases are trivial.

\begin{prop} [$\Az$ Feasibility] \label{prop:azex} 
    \eqref{eq:framework-a0} has a feasible solution if and only if $\left| \Om \right| \geq 3$.
\end{prop}
\begin{proof}
    Pick any $\Ad \in \Ac$. Consider the $\sNE$ condition. For any $a_1 \in \Ac_1$ and $a_2 \in \Ac_2$ such that $a_1 \neq \ad_1$ and $a_2 \neq \ad_2$, 
    \begin{align}
        \RR\left(a_1, \ad_2\right) < \RR\left(\ad_1, \ad_2\right) < \RR\left(\ad_1, a_2\right). \nonumber
    \end{align}
    The inequalities dictate that we need three different values to satisfy the condition for a fixed $a_1 \in \Ac_1$ and $a_2 \in \Ac_2$. If we know there exist three distinct $\omega_1 < \omega_2 < \omega_3 \in \Om$, then we can always construct the following feasible solution to \eqref{eq:framework-a0}.
    \begin{center}
    \begin{tabular}{c | c c c c}
        & $\ad_2$ & $a_2^{\left(1\right)}$ & $\cdots$ & $a_2^{\left(|\Ac_2|-1\right)}$  \\
        \hline
        $\ad_1$ & $\omega_2$ & $\omega_3$ & $\cdots$ & $\omega_3$ \\
        $a_1^{\left(1\right)}$ & $\omega_1$ & $\omega_2$ & $\cdots$ & $\omega_2$ \\
        $\vdots$ & $\vdots$ & $\vdots$ & & $\vdots$ \\
        $a_1^{\left(|\Ac_1|-1\right)}$ & $\omega_1$ & $\omega_2$ & $\cdots$ & $\omega_2$ \\
    \end{tabular}
    \end{center}
\end{proof}

The most important advantage of our approach is its strength. With the reward matrix constructed by our algorithm, the game changer can simultaneously control players seeking a wide variety of solution concepts beyond those that seek Nash equilibria. Notably, the game changer may control players that follow strict Nash Equilibria (sNE), pure Nash Equilibria (pNE), Nash Equilibria (NE), Correlated Equilibria (CE), and Coarse Correlated Equilibria (CCE). This stems from the fact that $\Ad \in \Ac$ is a strict Nash equilibrium of a finite two-player zero-sum game if and only if it is also the unique NE, CE, and CCE. 

\begin{lem} [$\Az$ Implications for Game Changer Feasibility and Optimality] \label{prop:azfs} 
    If $\RR$ is a feasible solution to ~\eqref{eq:framework-a0}, then $\RR$ is a feasible solution to ~\eqref{eq:fpdgc} for $S \in \left\{\sNE, \pNE, \NE, \CE, \CCE\right\}$. Furthermore, if $\RR$ is an optimal solution to ~\eqref{eq:framework-a0}, then $\RR$ is an optimal solution to ~\eqref{eq:fpdgc} for $S \in \left\{\sNE, \pNE, \NE, \CE, \CCE\right\}$.
\end{lem}

\begin{proof}
In zero-sum games, CE and CCE cannot improve beyond the minimax value, and if a pure action profile is a strict saddle point, all equilibrium distributions must put probability one on it. Conversely, since all NEs are also CEs and CCEs, $\ad$ being a unique CE or CCE implies that it is also a unique NE.
\end{proof}

In general, the optimal solution of \eqref{eq:framework-a0} is not necessarily unique. Consider the following example. 
\begin{align}
    & \RRo = \begin{bmatrix} 0 & 2 \\ 0 & 0 \end{bmatrix} &
    & \Om = \left\{-1, 0, 1, 2\right\}
    \nonumber
\end{align}
We can construct two different optimal solutions for \eqref{eq:framework-a0} using these inputs, namely
\begin{align}
    & \RR_1 = \begin{bmatrix} 1 & 2 \\ 0 & 0 \end{bmatrix} &
    & \RR_2 = \begin{bmatrix} 0 & 2 \\ -1 & 0 \end{bmatrix}
    \nonumber
\end{align}
Note that both modifications implement the same strict Nash equilibrium with an $L_1$ cost of $1$. 

\subsection{Algorithm $\Azdp$} \label{sec:alg-a0}

As formerly mentioned, one approach to solving \eqref{eq:framework-a0} is to apply integer programming. This is done by using binary-valued indicator variables to index into $\Om$. The details of this approach can be found in Appendix \ref{sec:alg-a0-ip}. 

We formulate a more efficient algorithm by focusing on the inner maximization in \eqref{eq:framework-a0}, as shown below. 
\begin{align}
    \min_{\RR} &\, \left\|\RR - \RRo\right\|_1 \label{eq:framework-a0-inner} \\
    \text{s.t.} 
    &\,\RR\left(a_1, \ad_2\right) < \RR\left(\ad_1, \ad_2\right) \quad \forall a_1 \neq \ad_1 \nonumber \\
    &\RR\left(\ad_1, a_2\right) > \RR\left(\ad_1, \ad_2\right) \quad \forall a_2 \neq \ad_2 \nonumber\\
    &\, \RR\left(\ac\right) \in \Om \quad \forall \ac \in \Ac \nonumber
\end{align}
We focus our attention on developing an efficient algorithm for solving \eqref{eq:framework-a0-inner}. Solving \eqref{eq:framework-a0} can be done by repeating our efficient algorithm to solve \eqref{eq:framework-a0-inner} for each $\Ad \in \Ac$, picking any solution with minimal cost. 

For convenience, we denote the sorted entries of $\Om$ by $\omega_1, \omega_2, \dots, \omega_m$, where $m = \card{\Om}$. If $\Om$ is not given in sorted order, then sorting these entries incurs a time complexity of $\BigOh{\card{\Om} \log \card{\Om}}$. We also define sentinels $\omega_0 = -\infty$ and $\omega_{m+1} = +\infty$. 

\newcommand{\Round}{\textsc{Round}}
Let us start with entries $\ac \in \Ac$ where $a_1 \neq \ad_1$ and $a_2 \neq \ad_2$. The only constraint on $\RR\left(\ac\right)$ is the $\Om$ constraint. We can efficiently pick $\RR\left(\ac\right)$ by finding an element $\omega_t \in \Om$ that minimizes $\left|\omega_t - \RRo\left(\ac\right)\right|$. To break ties, we arbitrarily choose to pick the smallest such $\omega_t$. We introduce a subroutine $\Round$ that performs this calculation using binary search over $\Om$. 
\begin{align}
    \RR\left(\ac\right) = \Round\left(\RRo\left(\ac\right), \Om\right)
    \label{eq:alg0-offdiag}
\end{align}
Intuitively, we approximate (or round) $\RRo\left(\ac\right)$ using values from $\Om$. 

The key decision that controls our choices for the remaining entries is how to set $\RR\left(\Ad\right)$. Suppose we choose $\RR\left(\Ad\right) = \omega_k$ for some $\omega_k \in \Om$. Consider some $\ac \in \Ac$ where $a_1 \neq \ad_1$ and $a_2 = \ad_2$, i.e., the row player does not follow $\Ad$ but the column player does. How should we pick $\RR\left(\ac\right)$? Ideally, to minimize the $L_1$ cost, we want to set $\RR\left(\ac\right)$ to $\omega_t = \Round\left(\RRo\left(\ac\right), \Om\right)$. However, we have the inequality constraint
\begin{align}
    \RR\left(\ac\right) < \RR\left(\Ad\right) = \omega_k
    \label{eq:alg0-ineq}
\end{align}
acting upon $\RR\left(\ac\right)$. If $\omega_t < \omega_k$, then no further action is necessary, but if $\omega_t \geq \omega_k$, then we cannot set $\RR\left(\ac\right)$ to $\omega_t$. In these cases, because $\omega_t$ minimizes $\left|\omega_t - \RRo\left(\ac\right)\right|$, we can guarantee that $\omega_{k-1}$ is the optimal choice for $\RR\left(\ac\right)$. As such, after fixing $\RR\left(\Ad\right) = \omega_k$, the optimal choice for $\RR\left(\ac\right)$ is
\begin{align}
    & \RR\left(\ac\right) = \min\left\{\omega_{k-1}, \omega_t\right\} &
    & \text{where $\omega_t = \Round\left(\RRo\left(\ac\right), \Om\right)$}.
    \label{eq:alg0-col}
\end{align}
The symmetric argument holds for any $\ac \in \Ac$ where $a_1 = \ad_1$ and $a_2 \neq \ad_2$, i.e., only the row player follows $\Ad$. If we choose $\RR\left(\Ad\right) = \omega_k \in \Om$, then the optimal choice for $\RR\left(\ac\right)$ is 
\begin{align}
    & \RR\left(\ac\right) = \max\left\{\omega_{k+1}, \omega_t\right\} & 
    \text{where $\omega_t = \Round\left(\RRo\left(\ac\right), \Om\right)$}.
    \label{eq:alg0-row}
\end{align}
With equations \eqref{eq:alg0-offdiag}, \eqref{eq:alg0-col}, and \eqref{eq:alg0-row}, we see why the choice of $\RR\left(\ad\right) = \omega_k \in \Om$ is important. Once our choice for $\omega_k$ is fixed, we have a closed form equations providing the optimal setting for the remaining entries of $\RR$. Algorithm~\ref{alg:alg0-slow} naturally follows. We iterate through every choice of $\omega_k$ and construct the corresponding optimal $\RR$. After every $\omega_k$ is considered, we pick whichever choice attained the optimal cost. 

\begin{algorithm}
    \caption{Initial algorithm for solving \eqref{eq:framework-a0-inner}}
    \label{alg:alg0-slow}
    \begin{algorithmic}
        \Require $\RRo$, $\Om$, $\ad$
        \ForAll{$\omega_k \in \Om$}
            \State $\RR\left(\Ad\right) \gets \omega_k$
            \ForAll{$\ac \in \Ac$}
                \If{$a_1 \neq \ad_1$ and $a_2 \neq \ad_2$}
                    \State $\RR\left(\ac\right) \gets \Round\left(\RRo\left(\ac\right), \Om\right)$
                \ElsIf{$a_1 \neq \ad_1$}
                    \State $\RR\left(\ac\right) \gets \min\left\{\omega_{k-1}, \Round\left(\RRo\left(\ac\right), \Om\right)\right\}$
                \ElsIf{$a_2 \neq \ad_2$}
                    \State $\RR\left(\ac\right) \gets \max\left\{\omega_{k+1}, \Round\left(\RRo\left(\ac\right), \Om\right)\right\}$
                \EndIf
            \EndFor
            \State $C_k \gets \left\|\RR - \RRo\right\|_1$
        \EndFor
        \State \Return $\RR$ 
    \end{algorithmic}
\end{algorithm}

This algorithm correctly solves \eqref{eq:framework-a0-inner}. We can build the optimal $\RR$ by extracting an $\omega_k$ that attained the optimal objective value, setting $\RR\left(\Ad\right) = \omega_k$, and reconstructing the remaining entries of $\RR$. However, this algorithm is quite slow. 

The primary overhead of this algorithm is reconstructing $\RR$ for every $\omega_k \in \Om$. We can optimize our algorithm by avoiding duplicate work. Consider entries $\ac \in \Ac$ where $a_1 \neq \ad_1$ and $a_2 \neq \ad_2$. Recall that $\RR\left(\ac\right)$ is not affected by our choice for $\RR\left(\Ad\right)$. This means that the calculations are the same for every $\omega_k \in \Om$. We optimize our algorithm by computing these entries once at the start of our algorithm. We need not consider these entries again. 

\newcommand{\SubA}{\alpha}
We continue by measuring the induced cost for setting every $\ac \in \Ac$ where $a_1 \neq \ad_1$ and $a_2 = \ad_2$, i.e., the row player does not follow $\Ad$ but the column player does. For convenience, we define $\SubA$ to be the set of all such action profiles. 
\begin{align}
    \SubA = \left\{\text{$\ac \in \Ac$ s.t. $a_1 \neq \ad_1$ and $a_2 = \ad_2$}\right\}
\end{align}
Ideally, we want to set each $\RR\left(\ac\right)$ to $\Round\left(\RRo\left(\ac\right), \Om\right)$, so we start by computing these values and measuring the induced $L_1$ cost $c$. We also keep some extra data to recall these calculations. For each $\omega_t \in \Om$, we record a list $R_t$ of all $\RRo\left(\ac\right)$ such that $\Round\left(\RRo\left(\ac\right), \Om\right) = \omega_t$. 
\begin{align}
    R_t = \left[\text{$\RRo\left(\ac\right)$ s.t. $\ac \in \SubA$ and $\Round\left(\RRo\left(\ac\right), \Om\right) = \omega_t$}\right]
\end{align}
Our strategy to avoid duplicate calculations is to maintain optimality while pursuing feasibility. Starting with this optimal configuration, we iterate through each $\omega_k \in \Om$. In each iteration, we adjust $\RR$ to attain feasibility for $\RR\left(\Ad\right) = \omega_k$. By ensuring optimality is maintained, we can efficiently compute the induced cost for each choice of $\omega_k$. 

We turn to satisfying \eqref{eq:alg0-ineq}. This inequality is less restrictive for larger choices of $\omega_k$. Our idea is to iterate through $\omega_k \in \Om$ in decreasing order, starting with the least restrictive $\omega_k = \omega_m$. To make our solution feasible, an entry $\RR\left(\ac\right)$ requires adjusting only if $\RR\left(\ac\right) = \omega_m$. These correspond to the elements of $R_m$. We manually adjust each $\RRo\left(\ac\right) \in R_m$ by setting $\RR\left(\ac\right) \gets \omega_{m-1}$ and recomputing their induced costs. Upon completion, we have an optimal configuration for $\omega_k = \omega_m$, so we record the optimal cost $c_k$. 

Continue with $\omega_k = \omega_{m-1}$. To make our solution feasible, an entry $\RR\left(\ac\right)$ requires adjusting if $\RR\left(\ac\right) = \omega_{m-1}$. As with the previous iteration, for each $\RRo\left(\ac\right) \in R_{m-1}$, we adjust $\RR\left(\ac\right) \gets \omega_{m-2}$ and recompute their induced costs. However, this iteration requires some more consideration. With our new choice of $\omega_k$, the entries $\RRo\left(\ac\right) \in R_m$ need to be readjusted. These entries were set to $\RR\left(\ac\right) = \omega_{m-1}$ in the previous iteration, so their adjustments will be identical. Rather than manually adjusting the cost for each of these entries, we can simply account for the total adjustment cost by setting
\begin{align}
    & c \gets c + \mu \cdot \left(\omega_k - \omega_{k-1}\right) &
    & \text{where $\mu$ is the size of $R_m$}.
\end{align}

\newcommand{\ComputeInducedCost}{\textsc{ComputeInducedCost}}
This pattern follows as we continue our process. For each $\omega_k \in \Om$, we manually adjust entries $\RRo\left(\ac\right) \in R_k$ and compute the changes in the induced cost. Then we count the number of entries that have been updated in previous iterations and efficiently adjust the cost using multiplication. Algorithm~\ref{alg:induced-cost} provides pseudocode for a subroutine $\ComputeInducedCost$ that follows this process. 

\begin{algorithm}
    \caption{Subroutine $\ComputeInducedCost$}
    \label{alg:induced-cost}
    \begin{algorithmic}
        \Procedure{ComputeInducedCost}{$\RRo$, $\Om$, $\SubA$}
            \ForAll{$\omega_t \in \Om$}
                \State $R_t \gets \left[\right]$
            \EndFor
            \State $c \gets 0$
            \ForAll{$\ac \in \SubA$}
                \State $\omega_t \gets \Round\left(\RRo\left(\ac\right), \Om\right)$
                \State Append $\RRo\left(\ac\right)$ to $R_t$
                \State $c \gets c + \left|\RRo\left(\ac\right) - \omega_t\right|$
            \EndFor
            \State $\mu \gets 0$
            \ForAll{$k \in \left[m, \dots, 3, 2, 1\right]$}
                \State $c \gets c + \mu \cdot \left(\omega_k - \omega_{k-1}\right)$
                \ForAll{$\RRo\left(\ac\right) \in R_k$}
                    \State $c \gets c - \left|\RRo\left(\ac\right) - \omega_k\right| + \left|\RRo\left(\ac\right) - \omega_{k-1}\right|$
                    \State $\mu \gets \mu + 1$
                \EndFor
                \State $c_k \gets c$
            \EndFor
            \State \Return $c_k$ for each $\omega_k \in \Om$
        \EndProcedure
    \end{algorithmic}
\end{algorithm}

\begin{lem}
	Subroutine $\ComputeInducedCost$ with inputs $\RRo$, $\Om$, and $\SubA$ runs with time complexity $\BigOh{\card{\Om} + \card{\SubA}\log\card{\Om}}$. 
\end{lem}
\begin{proof}
    We examine the pseudocode of Subroutine $\ComputeInducedCost$. The first overhead is the initial rounding phase. The subsequent iteration considers each element of $\Om$ exactly once. Moreover, each $\ac \in \SubA$ is considered exactly once. This puts the total time complexity at
    \begin{align*}
        \BigOh{\card{\SubA}\log\card{\Om}}
        + \BigOh{\card{\Om}}
        + \BigOh{\card{\SubA}}
        = \BigOh{\card{\Om} + \card{\SubA}\log\card{\Om}}
    \end{align*}
\end{proof}

\newcommand{\SubB}{\beta}
By symmetry, we can use subroutine $\ComputeInducedCost$ to compute the induced cost for setting every $\ac \in \Ac$ where $a_1 = \ad_1$ and $a_2 \neq \ad_2$, i.e., the row player does follows $\Ad$ but the column player does not. For convenience, we define
\begin{align}
    \SubB = \left\{\text{$\ac \in \Ac$ s.t. $a_1 = \ad_1$ and $a_2 \neq \ad_2$}\right\}
\end{align}
Our subroutine is designed to optimize \eqref{eq:alg0-col}, but we can adapt it for \eqref{eq:alg0-row} by negating the entries in $\RRo$ and $\Om$. With every profile $\ac \in \SubA \cup \SubB$ considered, the only remaining entry is $\Ad$. The cost for setting this single entry is trivially
\begin{align}
    \left|\RRo\left(\Ad\right) - \omega_k\right|
\end{align}
for each $\omega_k \in \Om$. With the cost of every entry considered, we construct Algorithm~\ref{alg:alg0}. 

\begin{algorithm}
    \caption{Algorithm $\Azdp$}
    \label{alg:alg0}
    \begin{algorithmic}
        \Procedure{$\Azdp$}{$\RRo$, $\Om$, $\Ad$}
            \State Sort $\Om$.
            \ForAll{$\omega_k \in \Om$}
                \State $C_k \gets \left|\RRo\left(\Ad\right) - \omega_k\right|$
            \EndFor
            \State $\SubA \gets \left\{\text{$\ac \in \Ac$ s.t. $a_1 \neq \ad_1$ and $a_2 = \ad_2$}\right\}$
            \ForAll{$c_k \in \ComputeInducedCost\left(\RRo, \Om, \SubA\right)$}
                \State $C_k \gets C_k + c_k$
            \EndFor
            \State $\SubB \gets \left\{\text{$\ac \in \Ac$ s.t. $a_1 = \ad_1$ and $a_2 \neq \ad_2$}\right\}$
            \ForAll{$c_k \in \ComputeInducedCost\left(-\RRo, -\Om, \SubB\right)$}
                \State $C_k \gets C_k + c_k$
            \EndFor
            \State $k \gets \arg \min_k C_k$
            \State $\RR\left(\Ad\right) \gets \omega_k$
            \ForAll{$\ac \in \Ac$}
                \If{$a_1 \neq \ad_1$ and $a_2 \neq \ad_2$}
                    \State $\RR\left(\ac\right) \gets \Round\left(\RRo\left(\ac\right), \Om\right)$
                \ElsIf{$a_1 \neq \ad_1$}
                    \State $\RR\left(\ac\right) \gets \min\left\{\omega_{k-1}, \Round\left(\RRo\left(\ac\right), \Om\right)\right\}$
                \ElsIf{$a_2 \neq \ad_2$}
                    \State $\RR\left(\ac\right) \gets \max\left\{\omega_{k+1}, \Round\left(\RRo\left(\ac\right), \Om\right)\right\}$
                \EndIf
            \EndFor
            \State \Return $\RR$
        \EndProcedure
    \end{algorithmic}
\end{algorithm}

\begin{thm}
	Algorithm $\Azdp$ correctly solves \eqref{eq:framework-a0-inner} and runs with time complexity
	\begin{align*}
        \BigOh{\left(\card{\Om} + \card{\Ac}\right) \cdot \log\card{\Om}}. 
	\end{align*}
	Moreover, the optimization problem \eqref{eq:framework-a0} can be solved with a time complexity of
	\begin{align*}
		\BigOh{\left(\card{\Om} + \card{\Ac}\right) \cdot \card{\Ac}\log\card{\Om}}. 
	\end{align*}
    \label{thm:alg0}
\end{thm}
\begin{proof}
    We walk through the pseudocode for Algorithm $\Azdp$. Notable workloads are sorting $\Om$, measuring costs using Subroutine $\ComputeInducedCost$, and constructing the final reward matrix. The total time complexity comes to
    \begin{align*}
        & \BigOh{\card{\Om}\log\card{\Om}}
        + \BigOh{\card{\Om} + \card{\Ac_1}\log\card{\Om} + \card{\Ac_2}\log\card{\Om}}
        + \BigOh{\card{\Ac}\log\card{\Om}} \\
        &= \BigOh{\left(\card{\Om} + \card{\Ac_1} + \card{\Ac_2} + \card{\Ac}\right) \cdot \log\card{\Om}}
        = \BigOh{\left(\card{\Om} + \card{\Ac}\right) \cdot \log\card{\Om}}
    \end{align*}
    We solve \eqref{eq:framework-a0} by running Algorithm $\Azdp$ for each $\Ad \in \Ac$ and picking the result with minimal cost, so the total time required to solve our problem is $\card{\Ac}$ times the time complexity of Algorithm $\Azdp$. 
\end{proof}

\section{Game Changer on General-sum Games}

We continue with general-sum games, permitting an arbitrary number of players. Here, we drop the zero-sum requirement on the output reward matrices (technically, it is a reward tensor/function). Note that a zero-sum game may be given as $\RRo$, but we will treat it as a general-sum game. Our strategy is to iterate through every pure strategy $\Ad \in \Ac$ and to install each $\Ad$ as a strict dominant strategy equilibrium. 

\begin{df} [Dominant Strategy Equilibrium] \label{def:sne} 
    Let $\RR$ be the reward function of a normal form game with action set $\Ac = \Ac_1 \times \dots \times \Ac_n$. A pure joint strategy $\Ad \in \Ac$ is a strict dominant strategy equilibrium ($\sDSE$) if for every player $i \in \ints{n}$ and every action $a_i \in \Ac_i$, we have
    \begin{align}
        a_i \neq \ad_i \implies \RR_i\left(a_i, a_{-i}\right) < \RR_i\left(\ad_i, a_{-i}\right), \forall a_{-i} \in \Ac_{-i}
        \nonumber.
    \end{align}
\end{df}

Mathematically, we are solving the following minimization equation. 
\begin{align}
    \max_{\Ad} \max_{\RR} &\, \Gval\left(\Ad\right) - \left\|\RR - \RRo\right\|_1 \label{eq:framework-a1} \\
    \text{s.t.} 
    &\, \RR_i\left(a_i, a_{-i}\right) < \RR_i\left(\ad_i, a_{-i}\right) \quad \forall i \in \ints{n}, \forall a_i \neq \ad_i, \forall a_{-i} \in \Ac_{-i} \nonumber\\
    &\, \RR\left(\ac\right) \in \Om \quad \forall \ac \in \Ac \nonumber
\end{align}

This approach shares several advantages with our zero-sum case. The only additional constraints are linear inequalities, allowing us to directly apply integer programming. We discuss an efficient closed-form algorithm in Section~\ref{sec:alg-a1}. We also have a similar sufficient and necessary condition for the existence of a feasible solution to \eqref{eq:framework-a1}. Again, we assume that each player has at least $2$ actions since degenerate one-action cases are trivial.

\begin{prop} [$\Ao$ Feasibility] \label{prop:aoex} 
    \eqref{eq:framework-a1} has a feasible solution if and only if $\left| \Om \right| \geq 2$.
\end{prop}
\begin{proof}
    We follow a similar proof by construction. Take $\omega_1 < \omega_2 \in \Om$ and pick any $\Ad \in \Ac$. Construct the following reward matrix. 
    \begin{align}
        & \forall \ac \in \Ac \quad \forall i \in \ints{n} &
        & \RR_i\left(\ac\right) = \begin{cases}
            \omega_2 & \text{if $a_i = \ad_i$} \\
            \omega_1 & \text{otherwise}
        \end{cases}
        \nonumber
    \end{align}
\end{proof}

The strength advantage for our general sum case requires extra consideration. Without the zero-sum restriction, it is difficult to determine whether a reward matrix has a unique Nash equilibrium. We choose to install a strict dominant strategy equilibrium because it guarantees control over players seeking a wide variety of solution concepts beyond those that seek Nash equilibria. Notably, the game changer may control players that follow Dominant Strategy Equilibria (DSE), weak Dominant Strategy Equilibria (wDSE), Nash equilibria, correlated equilibria, and coarse correlated equilibria. However, the dominant strategy equilibrium condition is not necessary for these equilibria to be unique.

\begin{lem} [$\Ao$ Implications for Game Changer Feasibility and Optimality] \label{prop:aofs} 
    If $\RR$ is a feasible solution to ~\eqref{eq:framework-a1}, then $\RR$ is a feasible solution to ~\eqref{eq:fpdgc} for $\SOL \in \left\{\sDSE, \DSE, \wDSE, \NE, \CE, \CCE\right\}$. Furthermore, if $\RR$ is an optimal solution to ~\eqref{eq:framework-a1}, then $\RR$ is an optimal solution to ~\eqref{eq:fpdgc} for $\SOL = \sDSE$, but not necessarily the other solution concepts.
\end{lem}

\begin{proof}
A unique sDSE implies that the action profile is also a unique DSE, wDSE, NE, CE, and CCE; but the converse is not necessarily true, for example, $\RR = \begin{bmatrix} (1, 1) & (0, 0) \\ (0, 3) & (0, 2) \end{bmatrix}$ has a unique NE that is not a sDSE.
\end{proof}

\subsection{Algorithm $\Aodp$} \label{sec:alg-a1}

As with the zero-sum case, one approach to solving \eqref{eq:framework-a1} is to apply integer programming. This is done by using binary-valued indicator variables to index into $\Om$. The details of this approach can be found in Appendix \ref{sec:alg-a1-ip}. 

We formulate a more efficient algorithm by focusing on the inner maximization in \eqref{eq:framework-a1}, as shown below. 
\begin{align}
    \min_{\RR} &\, \left\|\RR - \RRo\right\|_1 \label{eq:framework-a1-inner} \\
    \text{s.t.} 
    &\, a_i \neq \ad_i \implies \RR_i\left(a_i, a_{-i}\right) < \RR_i\left(\ad_i, a_{-i}\right) \quad \forall i \in \ints{n} \quad \forall a_i \in \Ac_i \\
    &\, \RR\left(\ac\right) \in \Om \quad \forall \ac \in \Ac \nonumber
\end{align}
We focus our attention on developing an efficient algorithm for solving \eqref{eq:framework-a1-inner}. Solving \eqref{eq:framework-a1} can be solved by repeating our efficient algorithm to solve \eqref{eq:framework-a1-inner} for each $\Ad \in \Ac$. 

Our efficient algorithm directly applies the $\ComputeInducedCost$ subroutine developed in Section~\ref{sec:alg-a0}. The structure of \eqref{eq:framework-a1-inner} implies strong separability. Namely, we can identify a separate subproblem for each pair of $i \in \ints{n}$ and $a_{-i} \in \Ac_{-i}$. This subproblem is exactly what $\ComputeInducedCost$ was designed to solve. We construct Algorithm~\ref{alg:alg1}. 

\begin{algorithm}
    \caption{Algorithm $\Aodp$}
    \label{alg:alg1}
    \begin{algorithmic}
        \Procedure{$\Aodp$}{$\RRo$, $\Om$, $\Ad$}
            \State Sort $\Om$.
            \ForAll{$i \in \ints{n}$}
                \ForAll{$a_{-i} \in \Ac_{-i}$}
                    \State $\SubA \gets \left\{\text{$\left(a_i, a_{-i}\right)$ s.t. $a_i \in \Ac_i$ and $a_i \neq \ad_i$}\right\}$
                    \ForAll{$c_k \in \ComputeInducedCost\left(\RRo_i, \Om, \SubA\right)$}
                        \State $C_k \gets \left|\RRo_i\left(\ad_i, a_{-i}\right) - \omega_k\right| + c_k$
                    \EndFor
                    \State $k \gets \arg\min_k C_k$
                    \State $\RR_i\left(\ad_i, a_{-i}\right) \gets \omega_k$
                    \ForAll{$a_i \in \Ac_i$}
                        \State $\RR_i\left(a_i, a_{-i}\right) \gets \min\left\{\omega_{k-1}, \Round\left(\RRo_i\left(a_i, a_{-i}\right), \Om\right)\right\}$
                    \EndFor
                \EndFor
            \EndFor
            \State \Return $\RR$
        \EndProcedure
    \end{algorithmic}
\end{algorithm}

We note that the time complexity of this algorithm is exponential in the number of players. However, this exponential nature is unavoidable because the action space has an exponential number of elements, and the reward function is $n$ times larger than the action space. 

\begin{thm}
	Algorithm $\Aodp$ correctly solves \eqref{eq:framework-a1-inner} and runs with time complexity
	\begin{align*}
		\BigOh{\left(\card{\Om} + n \card{\Ac}\right) \log\card{\Om} + \sum_{i=1}^n \card{\Ac_{-i}} \card{\Om}}. 
	\end{align*}
    Moreover, the optimization problem \eqref{eq:framework-a1} can be solved with a time complexity of
	\begin{align*}
		\BigOh{\card{\Ac} \cdot\left(\left(\card{\Om} + n \card{\Ac}\right) \log\card{\Om} + \sum_{i=1}^n \card{\Ac_{-i}} \card{\Om}\right)}. 
	\end{align*}
\end{thm}
\begin{proof}
    We start with sorting $\Om$. Then for each player $i \in \ints{n}$ and each $a_{-i} \in \Ac_{-i}$, our primary overheads are running Subroutine $\ComputeInducedCost$ and rounding entries. Putting these together and taking the summation over $\ints{n}$ and $\Ac_{-i}$, we get
    \begin{align*}
        & \BigOh{\card{\Om}\log\card{\Om}}
        + \BigOh{\sum_{i=1}^n \card{\Ac_{-i}} \cdot \left(\card{\Om} + \card{\Ac_i}\log\card{\Om}\right)} \\
        &= \BigOh{\card{\Om}\log\card{\Om} +
        \sum_{i=1}^n \card{\Ac_{-i}} \cdot\card{\Om} + \sum_{i=1}^n \card{\Ac}\log\card{\Om}} \\
        &= \BigOh{\left(\card{\Om} + n \card{\Ac}\right)\log\card{\Om} + \sum_{i=1}^n \card{\Ac_{-i}} \cdot\card{\Om}}.
    \end{align*}
    We solve \eqref{eq:framework-a1} by repeating Algorithm $\Aodp$ for each $\ac \in \Ac$, so the total time required is $\card{\Ac}$ times the time complexity of Algorithm $\Aodp$. 
\end{proof}

\section{Experiments}

To provide extra evidence for the veracity of our algorithms, we perform empirical experiments. Our first experiment shows the correctness of our efficient approach. We generate several randomized games with sizes
\begin{align*}
    & \card{\Ac_1} = 6 &
    & \card{\Ac_2} = 6 &
    & \card{\Om} = 10.
\end{align*}
Entries for $\RRo$ and $\Om$ are picked uniformly at random on the range $\left[-10, 10\right]$. We also pick $\Ad$ uniformly at random from $\Ac$. These inputs are passed through Algorithm $\Azdp$, along with the integer programming approach outlined in Appendix \ref{sec:alg-a0-ip}. We apply support enumeration to identify all Nash equilibria in the outputs of our algorithms. Every trial yielded $\Ad$ as the unique Nash equilibrium. Moreover, the optimal modification costs attained by our two approaches were always consistent within an error of $10^{-7}$ attributed to float imprecision. This supports our claim that our algorithms provide correct solutions. 

Next, we check the runtime complexity by scaling $\card{\Ac}$ and $\card{\Om}$. Here, we generate our input data as described above, ensuring to exclude data generation from our time measurements. For each input size, we perform $50$ repetitions of Algorithm $\Azdp$, measuring the total elapsed time. 

\begin{center}
\begin{tikzpicture}
\begin{axis}[
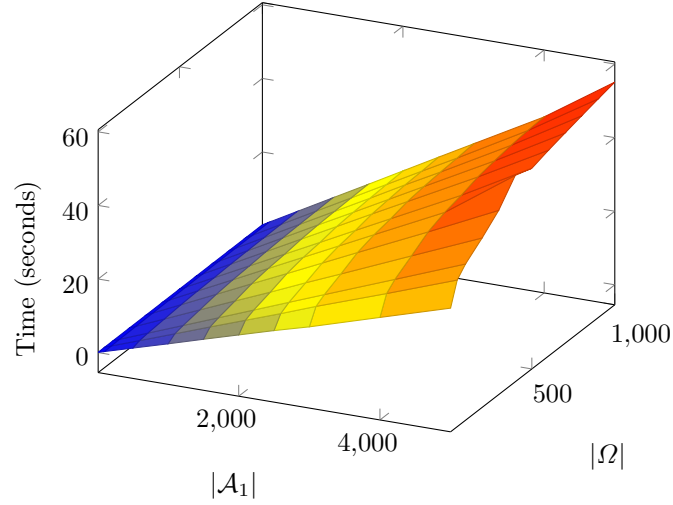
    title={Runtime for $50$ iterations of Algorithm $\Azdp$ ($\card{\Ac_2} = 100$)},
    xlabel=$\card{\Ac_1}$,
    ylabel=$\card{\Om}$,
    zlabel={Time (seconds)}
]
\addplot3[
    surf
] 
coordinates {
(10, 10, 0.06728601455688477) (10, 50, 0.08365726470947266) (10, 100, 0.09328532218933105) (10, 200, 0.10855770111083984) (10, 300, 0.12537622451782227) (10, 400, 0.13810014724731445) (10, 500, 0.14896392822265625) (10, 600, 0.16190767288208008) (10, 700, 0.173431396484375) (10, 800, 0.1821601390838623) (10, 900, 0.19285273551940918) (10, 1000, 0.20063161849975586)

(100, 10, 0.5774922370910645) (100, 50, 0.7002785205841064) (100, 100, 0.7589352130889893) (100, 200, 0.8401086330413818) (100, 300, 0.9226458072662354) (100, 400, 0.9737792015075684) (100, 500, 1.0132677555084229) (100, 600, 1.0588483810424805) (100, 700, 1.095005750656128) (100, 800, 1.1193158626556396) (100, 900, 1.1291849613189697) (100, 1000, 1.1520380973815918)

(500, 10, 2.8005988597869873) (500, 50, 3.4391207695007324) (500, 100, 3.7220571041107178) (500, 200, 4.0802788734436035) (500, 300, 4.471395254135132) (500, 400, 4.721101522445679) (500, 500, 4.856089353561401) (500, 600, 5.062113523483276) (500, 700, 5.175363302230835) (500, 800, 5.254856109619141) (500, 900, 5.334980726242065) (500, 1000, 5.390238285064697)

(1000, 10, 5.57322883605957) (1000, 50, 6.90128231048584) (1000, 100, 7.4184465408325195) (1000, 200, 8.16275668144226) (1000, 300, 8.919299602508545) (1000, 400, 9.387130737304688) (1000, 500, 9.670968770980835) (1000, 600, 10.062345027923584) (1000, 700, 10.292425870895386) (1000, 800, 10.451161623001099) (1000, 900, 10.565197467803955) (1000, 1000, 10.674848079681396)

(1500, 10, 8.424988508224487) (1500, 50, 10.383661031723022) (1500, 100, 11.120664834976196) (1500, 200, 12.21170711517334) (1500, 300, 13.38481855392456) (1500, 400, 14.070987701416016) (1500, 500, 14.48316764831543) (1500, 600, 15.089397430419922) (1500, 700, 15.433679580688477) (1500, 800, 15.633321523666382) (1500, 900, 15.792678833007812) (1500, 1000, 15.945148944854736)

(2000, 10, 11.220191717147827) (2000, 50, 13.812153816223145) (2000, 100, 14.838740110397339) (2000, 200, 16.219172954559326) (2000, 300, 17.780407667160034) (2000, 400, 18.73727512359619) (2000, 500, 19.288519859313965) (2000, 600, 20.100563287734985) (2000, 700, 20.488830089569092) (2000, 800, 20.80468511581421) (2000, 900, 21.066144943237305) (2000, 1000, 21.272650241851807)

(2500, 10, 14.016608715057373) (2500, 50, 17.232468366622925) (2500, 100, 18.562923908233643) (2500, 200, 20.378288745880127) (2500, 300, 22.251075267791748) (2500, 400, 23.409446954727173) (2500, 500, 24.143051385879517) (2500, 600, 25.139069080352783) (2500, 700, 25.59890842437744) (2500, 800, 26.081085920333862) (2500, 900, 26.300424098968506) (2500, 1000, 26.588749647140503)

(3000, 10, 16.644092082977295) (3000, 50, 20.806584119796753) (3000, 100, 22.264296293258667) (3000, 200, 24.39317226409912) (3000, 300, 26.7255117893219) (3000, 400, 28.093962907791138) (3000, 500, 28.943512439727783) (3000, 600, 30.103095293045044) (3000, 700, 30.717108488082886) (3000, 800, 31.160414218902588) (3000, 900, 31.471003770828247) (3000, 1000, 31.864861011505127)

(4000, 10, 22.432027339935303) (4000, 50, 27.530016899108887) (4000, 100, 29.6343994140625) (4000, 200, 32.535956382751465) (4000, 300, 35.555272579193115) (4000, 400, 37.38176155090332) (4000, 500, 38.53745675086975) (4000, 600, 40.038082122802734) (4000, 700, 40.91953730583191) (4000, 800, 41.65477633476257) (4000, 900, 42.162832260131836) (4000, 1000, 42.40646481513977)

(5000, 10, 28.182419300079346) (5000, 50, 34.426724672317505) (5000, 100, 37.044355630874634) (5000, 200, 40.624919414520264) (5000, 300, 44.64346885681152) (5000, 400, 50.67576503753662) (5000, 500, 49.1145122051239) (5000, 600, 50.522528886795044) (5000, 700, 51.721699237823486) (5000, 800, 52.90037941932678) (5000, 900, 54.024823904037476) (5000, 1000, 55.09347605705261)
};
\end{axis}
\end{tikzpicture}
\end{center}
Our data demonstrates the expected linear dependence on $\card{\Ac}$, along with a logarithmic dependence on $\card{\Om}$. We note that the logarithmic dependence on $\card{\Om}$ is consistent with Theorem~\ref{thm:alg0} because our action space is large. We perform an additional experiment to show that as $\card{\Om}$ becomes much larger than $\card{\Ac}$, the $\BigOh{\card{\Om}\log\card{\Om}}$ pattern arises. 
\begin{center}
\begin{tikzpicture}
\begin{axis}[
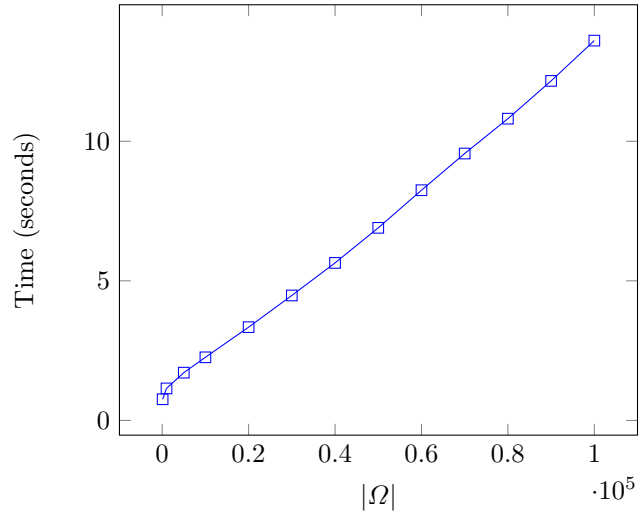
    title={Runtime for $50$ repetitions of Algorithm $\Azdp$ using large $\card{\Om}$},
    xlabel={$\card{\Om}$},
    ylabel={Time (seconds)}
]

\addplot[
    color=blue,
    mark=square,
    ]
    coordinates {
    (100, 0.7571992874145508) (1000, 1.145888328552246) (5000, 1.710533618927002) (10000, 2.260711193084717) (20000, 3.3408608436584473) (30000, 4.476373910903931) (40000, 5.642657995223999) (50000, 6.899082183837891) (60000, 8.247092485427856) (70000, 9.5615394115448) (80000, 10.809693574905396) (90000, 12.163147211074829) (100000, 13.606081008911133)
    };
\end{axis}
\end{tikzpicture}
\end{center}

\section{Appendix}

\subsection{$\Azip$ with Integer Programs} \label{sec:alg-a0-ip}

We reformulate \eqref{eq:framework-a0-inner} as a binary integer program by defining one-hot vectors as decision variables for each $\ac \in \Ac$. Let $\tilde{\RR}\left(\ac\right)$ be a vector of $\left|\Om\right|$ binary decision variables. To ensure we have a one-hot decision vector, we add the constraint that $1^\top_{\left| \Om \right|} \tilde{\RR}\left(\ac\right) = 1$. We can then construct the reward matrix $\RR$ as follows. 
\begin{align*}
    & \RR\left(\ac\right) = \left\langle\tilde{\RR}\left(\ac\right), \omega\right\rangle &
    & \omega = \begin{bmatrix}
        \omega_1 \\ \vdots \\ \omega_{\left|\Om\right|}
    \end{bmatrix}
\end{align*}
We can think of $\tilde{\RR}\left(\ac\right)$ as an indicator function for the condition $\RR\left(\ac\right) = \omega_i$. 

The remaining work needed to formulate the binary integer program is to resolve the strict inequalities from the $\sNE$ constraint. We do this by relaxing the strict inequalities into weak inequalities with a gap of $\iota > 0$. 
\begin{align}
    \min_\RR & \left\| \RR - \RRo \right\|_1 \label{eq:azlp} \\
    \text{s.t. }
    & \RR\left(a_1, \ad_2\right) + \iota \leq \RR\left(\ad_1, \ad_2\right) \leq \RR\left(\ad_1, a_2\right) - \iota \nonumber \\
    & \quad \forall a_1 \in \Ac_1 \quad \forall a_2 \in \Ac_2 \quad \text{s.t. $a_1 \neq \ad_1$ and $a_2 \neq \ad_2$} \nonumber \\
    & \RR\left(a\right) = \left\langle\tilde{\RR}\left(\ac\right), \omega\right\rangle, \forall\; \ac \in \Ac \nonumber \\
    & 1^\top_{\left| \Om \right|} \tilde{\RR}\left(\ac\right) = 1, \forall\; \ac \in \Ac,\nonumber \\
    & \tilde{\RR}\left(\ac\right) \in \left\{0, 1\right\}^{\left|\Om\right|}, \forall\; \ac \in \Ac \nonumber
\end{align}

\begin{rmk} \label{lem:azeq} 
\eqref{eq:azlp} is equivalent to \eqref{eq:framework-a0-inner} for every $\iota \leq \Delta$, where the value gap $\Delta$ is defined as
\begin{align}
\Delta = \min_{\omega \neq \omega' \in \Om} \left| \omega - \omega' \right|. \label{eq:dxz}
\end{align}\end{rmk}

\subsection{$\Aodp$ with Integer Programs} \label{sec:alg-a1-ip}
For $\Ao$, we have an integer program similar to ~\eqref{eq:azlp} with the $\iota$ relaxation, where the strict inequalities are replaced by weak inequalities with a gap $\iota > 0$.
\begin{align}
\displaystyle\min_{\tilde{\RR}} & \left\|\RR - \RRo\right\|_{1} \label{eq:aolp}
\\ \text{\;s.t.\;} & \omega^\top \tilde{\RR_i}\left(\ad_i, a_{-i}\right) \geq \omega^\top \tilde{\RR_i}\left(a_i, a_{-i}\right) + \iota,\nonumber
\\ & \hspace{2em} \forall\; a_i \neq \ad_i, \forall\; a_{-i} \in \Ac_{-i}, i \in \left[n\right],\nonumber
\\ & \RR\left(\ac\right) = \left\langle\tilde{\RR}\left(\ac\right), \omega\right\rangle, \forall\; \ac \in \Ac \nonumber
\\ & 1^\top_{\left| \Om \right|} \tilde{\RR}\left(\ac\right) = 1, \forall\; \ac \in \Ac,\nonumber
\\ & \tilde{\RR}\left(\ac\right) \in \left\{0, 1\right\}^{\left| \Om \right|}, \forall\; \ac \in \Ac.\nonumber
\end{align}

\begin{rmk} [$\Aoip$ Equivalence] \label{lem:aoeq} 
~\eqref{eq:aolp} is equivalent to \eqref{eq:framework-a1-inner} for every $\iota < \Delta$, with the same $\Delta$ defined in~\eqref{eq:dxz}.

\end{rmk}

\subsection{Integer Programs for Zero-sum Markov Games} \label{sec:alg-a0-ip-mg}

Binary integer programs can also be applied to Markov game redesign by adding the Q value function constraints. We present the following integer program formulation to modify the zero-sum Markov game with state space $S$, reward function $\RR$ and transition $\mathcal{P}$ so that $\ad$ becomes the unique Markov perfect equilibrium.

\begin{align}
    \min_\RR & \left\| \RR - \RRo \right\|_1 \label{eq:azlpmg} \\
    \text{s.t. }
    & \QQ\left(s, a_1, \ad_2 \left(s\right)\right) + \iota \leq \QQ\left(s, \ad_1 \left(s\right), \ad_2 \left(s\right)\right) \leq \QQ\left(s, \ad_1, a_2 \left(s\right)\right) - \iota \nonumber \\
    & \quad \forall s \in S \quad \forall a_1 \in \Ac_1 \quad \forall a_2 \in \Ac_2 \quad \text{s.t. $a_1 \neq \ad_1 \left(s\right)$ and $a_2 \neq \ad_2 \left(s\right)$} \nonumber \\
    & \QQ(s, a_1, a_2) = \RR(s, a_1, a_2) + \gamma \sum_{s' \in S} \mathcal{P}\left(s' | s, a_1, a_2\right) \QQ(s', \ad_1 \left(s'\right), \ad_2 \left(s'\right)) \nonumber \\
    & \quad \forall s \in S \quad \forall a_1 \in \Ac_1 \quad \forall a_2 \in \Ac_2 \nonumber \\
    & \RR\left(\ac\right) = \left\langle\tilde{\RR}\left(\ac\right), \omega\right\rangle, \forall\; \ac \in \Ac \nonumber \\
    & 1^\top_{\left| \Om \right|} \tilde{\RR}\left(\ac\right) = 1, \forall\; \ac \in \Ac,\nonumber \\
    & \tilde{\RR}\left(\ac\right) \in \left\{0, 1\right\}^{\left|\Om\right|}, \forall\; \ac \in \Ac \nonumber
\end{align}

Here, the $\iota \leq \Delta$ condition is not sufficient for~\eqref{eq:azlpmg} to produce the optimal solution; however, solving the problem for sufficiently small $\iota$ or a sequence of decreasing $\iota$'s will lead to the exact optimal solution.

\vspace{3em}
\begingroup
\let\clearpage\relax
\let\newpage\relax
\bibliography{dinash}
\endgroup

\end{document}